\documentclass[smallabstract,smallcaptions]{dccpaper}
\usepackage[utf8]{inputenc}

\usepackage{graphicx}
\usepackage[english]{babel}
\usepackage{amsmath}
\usepackage{amsthm}
\usepackage{amssymb}
\usepackage{forest}
\usepackage{color}
\usepackage{url}
\usepackage{tikz}
\usetikzlibrary{arrows,automata,positioning}
\usepackage{enumerate}
\usepackage{float}
\theoremstyle{definition}
\newtheorem{defi}{Definition}

\theoremstyle{plain}
\newtheorem{theorem}[defi]{Theorem}

\newtheorem{lemma}[defi]{Lemma}

\usepackage{epsfig}

\newlength{\figurewidth}
\newlength{\smallfigurewidth}

\title
{
\large
\textbf{Tunneling on Wheeler Graphs}
}

\author{%
Jarno Alanko$^{1}$, Travis Gagie$^{2,3}$, 
Gonzalo Navarro$^{2,4}$, Louisa Seelbach Benkner$^{5}$\\[0.5em]
{\small\begin{minipage}{\linewidth}\begin{center}
$^{1}$Dept. of Computer Science, University of Helsinki, Finland, jarno.alanko@helsinki.fi \\
$^{2}$CeBiB --- Center for Biotechnology and Bioengineering, Chile \\
$^3$EIT, Diego Portales University, Chile, travis.gagie@gmail.com \\
$^4$Dept. of Computer Science, University of Chile, Chile, gnavarro@dcc.uchile.cl \\
$^{5}$Dept. of Electrical Engineering and Computer Science, University of Siegen, Germany, seelbach@eti.uni-siegen.de \\
\end{center}\end{minipage}}
}

\begin{document}

\maketitle
\thispagestyle{empty}

\maketitle

\section*{Abstract}
% arxiv version
The Burrows-Wheeler Transform (BWT) is an important technique both in data compression and in the design of compact indexing data structures.  It has been generalized from single strings to collections of strings and some classes of labeled directed graphs, such as tries and de Bruijn graphs.  The BWTs of repetitive datasets are often compressible using run-length compression, but recently Baier (CPM 2018) described how they could be even further compressed using an idea he called {\em tunneling}.  In this paper we show that tunneled BWTs can still be used for indexing and extend tunneling to the BWTs of Wheeler graphs, a framework that includes all the generalizations mentioned above.
%Baier (CPM 2018) describes {\em tunneling} as a technique to further exploit redundancies in the Burrows-Wheeler Transform. In this paper we show how to retain indexed text searching on the resulting structure and generalize the concept to Wheeler graphs.

\Section{Introduction}

The Burrows-Wheeler transform (BWT) is a cornerstone of data compression and succinct text indexing. It is a reversible permutation on a string that tends to compress well with run-length coding, while simultaneously facilitating pattern matching against the original string. Recently, practical data structures have been designed on top of the run-length compressed BWT to support optimal-time text indexing within space bounded by the number of runs of the BWT \cite{gagie2018optimal}.

However, run-length coding does not necessarily exploit all the available redundancy in the BWT of a repetitive string. To this end, Baier recently introduced the concept of \emph{tunneling} to compress the BWT by exploiting additional redundancy not yet captured by run-length compression \cite{baier2018undetected}. While his representation can achieve better compression than run-length coding, no support for text indexing is given.

Meanwhile, the concept of \emph{Wheeler graphs} was introduced by Gagie et al.\ as an alternative way to view Burrows-Wheeler type indices \cite{gagie2017wheeler}. The framework can be used to derive a number of existing index structures in the Burrows-Wheeler family, like the classical FM-index \cite{FM05} including its variants for multiple strings \cite{mantaci2007extension} and alignments \cite{na2016fm}, the XBWT for trees \cite{ferragina2005structuring}, the GCSA for directed acyclic graphs \cite{siren2014indexing}, and the BOSS data structure for de Bruijn graphs \cite{bowe2012succinct}.

In this work, we show how Baier's concept of tunneling can be neatly explained in terms of Wheeler graphs. Using the new point of view, we show how to support FM-index style pattern searching on the tunneled BWT. We also describe a sampling strategy to support pattern counting and locating and character extraction, making our set of data structures a fully-functional FM-index. Supporting FM-index operations was posed as an open problem by Baier \cite{baier2018undetected}.

We also use the generality of the Wheeler graph framework to generalize the concept of tunneling to any Wheeler graph. This result can be used to compress any  Wheeler graph while still supporting basic pattern searching to decide if a pattern exists as a path label in the graph. This has applications for all index structures that can be explained in terms of Wheeler graphs.

\Section{Preliminaries}

Let $\mathcal{G}=(V, E, \lambda)$ denote a directed edge-labeled multi-graph, in which $V$ denotes the set of nodes, $E$ denotes the multiset of edges and $\lambda: E \rightarrow A$ denotes a function labeling each edge of $\mathcal{G}$ with a character from a totally-ordered alphabet $A$. Let $\prec$ denote the ordering among $A$'s elements. We follow the definition of Gagie et al.~\cite{gagie2017wheeler}.

\begin{defi}[Wheeler graph] \label{def:wg}
The graph $\mathcal{G}=(V, E, \lambda)$ is called a \emph{Wheeler graph} if there is an ordering on the set of nodes such that nodes with in-degree $0$ precede those with positive in-degree and for any two edges $(u_1,v_1), (u_2, v_2)$ labeled with $\lambda((u_1,v_1))=a_1 $ and $\lambda((u_2, v_2))=a_2$, we have
\begin{itemize}
\item[(i)] $a_1 \prec a_2 \Rightarrow v_1 < v_2$,
\item[(ii)] $\left(a_1=a_2 \right) \wedge \left( u_1 <  u_2 \right) \Rightarrow v_1 \leq v_2$. \qed
\end{itemize}
\end{defi}
We note that the definition implies that all edges arriving at a node have the same label.
We call such an ordering a \emph{Wheeler ordering of nodes}, and the rank of a node within this ordering the \emph{Wheeler rank} of the node. Given a pattern $P \in A^*$, we call a node $v \in V$ an \emph{occurrence} of $P$ if there is a path in $\mathcal G$ ending at $v$ such that the concatenation of edge labels on the path is equal to $P$. The Wheeler ranks of all occurrences of $P$ form a contiguous range of integers, which we call the $\emph{Wheeler range}$ of $P$. Finding such a range is called {\em path searching} for $P$.
% We also use the term Wheeler range to refer to any set of nodes such that the Wheeler ranks of the nodes form a contiguous interval.

Given a Wheeler ordering of nodes, we define the corresponding \emph{Wheeler ordering of edges} such that for a pair of edges $e_1 = (u_1,v_1)$ and $e_2 = (u_2, v_2)$, we have $e_1 < e_2$ iff $\lambda(e_1) \prec \lambda(e_2)$ or ($\lambda(e_1) = \lambda(e_2)$ and $u_1 < u_2)$.
When referring to edges, the term \emph{Wheeler rank} refers to the rank of the edge in the Wheeler ordering of edges, and \emph{Wheeler range} refers to a set of edges whose Wheeler ranks form a contiguous interval. 

We use a slightly modified version of the representation of Wheeler graphs  proposed by Gagie et al.~\cite{gagie2017wheeler}. Suppose we have a Wheeler graph with $n$ nodes and $m$ edges. Then we represent the graph with the following data structures:
\begin{itemize}
    \item A string $L[1..m] = L_1 \cdots L_{n}$ where $L_i$ is the concatenation of the labels of the $|L_i|$ edges going out from the node with Wheeler rank $i$ such that the labels from a node are concatenated in their relative Wheeler order.
    \item An array $C[1..|A|]$ such that $C[i]$ is the number of edges in $E$ with a label smaller than the $i$th smallest symbol in $A$.
    \item A binary string $I[1..n + m + 1] = X_1,\cdots,X_{n} \cdot 1$ where $X_i$ = $1 \cdot 0^{k_i}$ and $k_i$ is the indegree of the node with Wheeler rank $i$.
    \item A binary string $O[1..n + m + 1] = X_1,\cdots,X_{n} \cdot 1$ where $X_i$ = $1 \cdot 0^{l_i}$ and $l_i$ is the outdegree of the node with Wheeler rank $i$.
\end{itemize}
%Given these data structures, we can traverse the Wheeler graph with the following two operations: \textbf{(a)} given a Wheeler range $W$ of nodes, find the Wheeler range of edges with label $c$ whose origins are in $W$; \textbf{(b)} given a Wheeler range of edges, find the Wheeler range of nodes at the destinations of those edges \cite{gagie2017wheeler}.
Given these data structures, we can traverse the Wheeler graph with the following two operations: First, given the Wheeler rank $i$ of a node, find the Wheeler rank of the $k$-th out-edge labeled $c$ from the node using Eq.~(\ref{eq:node_to_edge}) below; second, given the Wheeler rank $j$ of an edge, find the Wheeler-rank $r$ of its target node with Eq.~(\ref{eq:edge_to_node}):
\begin{equation} \label{eq:node_to_edge}
C[c] + rank_c(L, select_1(O, i)-i) + k, 
\end{equation}
\begin{equation} \label{eq:edge_to_node}
rank_1(I, select_0(I, j)), % todo: simplify?
\end{equation}
where, given a string $S$, $rank_c(S,i)$ denotes the number of occurrences of character $c$ in $S[1..i]$ and $select_c(S,i)$ denotes the position of the $i$-th occurrence of $c$ in $S$. For the binary strings $I$ and $O$, these operations can be carried out in constant time using $o(|I|+|O|)$ extra bits \cite{Cla96}. For string $L$, operation $rank$ can be carried out in time $O(\log\log_w |A|)$ on a $w$-bit RAM machine using $o(|L|\log|A|)$ extra bits \cite[Thm.~8]{BNtalg14}.

Given a Wheeler range $[i,i']$ of nodes, we can find the Wheeler rank of the first and last edge labeled $c$ leaving from a node in $[i,i']$ with a variant of Eq.~(\ref{eq:node_to_edge}): $C[c]+rank_c(L,select_1(O,i)-i)+1$ and $C[c]+rank_c(L,select_1(O,i'+1)-(i'+1))$, respectively; then we apply Eq.~(\ref{eq:edge_to_node}) on both edge ranks to obtain the corresponding node range (the result is a range by the path coherence property of Wheeler graphs \cite{gagie2017wheeler}). This operation enables path searches on Wheeler graphs.

%Given a Wheeler range $W$ of nodes and a character $c$, we can find the Wheeler range of edges with label $c$ such that the origins of the edges are in $W$ ($select_1$ on $O$ and $rank_c$ on $L$). (b) Given a range of edges, we can find the range of nodes at the destinations of those edges ($rank_1$ on $I$ and $C$-array (i.e. select on $F$)).

\Section{Tunneling}

We adapt Baier's concept of \emph{blocks} on the BWT \cite{baier2018undetected} to Wheeler graphs.

\begin{defi}[Block]\label{def:block}
A \emph{block} $\mathcal{B}$ of a Wheeler graph $\mathcal{G}=(V, E, \lambda)$ of size $s$ and width $w$ is a sequence of $w$-tuples $(v_{1,1}, \dots, v_{w,1}), \dots, (v_{1,s}, \dots, v_{w,s})$ of pairwise distinct nodes of $\mathcal{G}$ such that
\begin{itemize}
\item[(i)] For  $1 \leq i \leq w-1$ and $1 \leq j \leq s$, the node $v_{i+1,j}$ is the immediate successor of $v_{i,j}$ with respect to the Wheeler ordering on $V$.
\item[(ii)] For $1 \leq i \leq w$, let $V_i=\{v_{i,j} \mid 1 \leq j \leq s\}$, $E_i= E \cap (V_i \times V_i)$, and $\lambda_i = \lambda \big|_{E_i}$. The subgraphs $t_i=(V_i, E_i, \lambda_i)$ are isomorphic subtrees of $\mathcal{G}$, preserving topology and labels.
For $1 \leq i \leq w-1$, let $f_i: t_i \rightarrow t_{i+1}$ denote the corresponding isomorphisms, thus $v_{i+1,j} = f_i(v_{i,j})$ for all $1 \le j \le s$.
\item[(iii)] For $1 \leq i \leq w$, let $v_{i,1}$ denote the root node of $t_i$. In particular, $v_{i,1}$ is the only node of indegree $0$ in $t_i$. All edges leading to a node in $\{v_{i,1} \mid 1 \leq i \leq w\}$ are labeled with the same character. The indegrees of these nodes may differ. 
\item[(iv)] For $1 \leq i \leq w$ and $2\leq j \leq s$, the nodes $v_{i,j}$ are of indegree $1$ in $\mathcal{G}$ (and by (i) and (ii), of indegree $1$ in the corresponding subtree $t_i$, that is, the only edge in $\mathcal{G}$ leading to such a node $v_{i,j}$ belongs to the subtree $t_i$).

\item[(v)] For every integer $1 \leq j \leq s$ and character $c \in A$, exactly one of the following conditions holds: 
\begin{itemize}
    \item[(a)] For every $1 \leq i \leq w$, there is exactly one out-edge of $v_{i,j}$ labeled with $c$, which is contained in $E_i$. There are no out-edges of $v_{i,j}$ labeled with $c$ leading to non-block nodes. % WARNING: THERE IS A HARD-CODED CROSS-REFERENCE TO THIS ITEM IN THE PATH SEARCHING SECTION. IF YOU UPDATE THIS DEFINITION, CHECK THAT THE REFERENCE IS STILL CORRECT
    \item[(b)] For every $1 \leq i \leq w$, there is no out-edge of $v_{i,j}$ labeled with $c$ contained in $E_i$. There may be out-edges of $v_{i,j}$ labeled with $c$ leading to non-block nodes. The number of such out-edges for each node may differ. \qed
\end{itemize}
\end{itemize}
\end{defi}

\definecolor{mygray}{rgb}{0.88, 0.88, 0.88}

\begin{figure}[t]
\tiny{
\begin{tikzpicture}[->,>=stealth',shorten >=1pt,auto,node distance=4em, semithick]
  \tikzstyle{every state}=[]

  \node[state, fill=mygray] (s1) [] {25};
    \node[state, fill=mygray] (s2) [below right of =s1]{32};
   \node[state, fill=mygray ](s3)[below left of =s1]{64};
   
   \node[state, fill=mygray](s4)[below of =s2]{10};
   \node[state, fill=mygray](s6)[ left of =s4]{42};
   \node[state, fill=mygray](s7)[left of=s6]{75};
   
   \node[state, fill=mygray] (s5)[below of =s4]{5};

   \node[state] (e1)[ above right of =s1]{};
   \node[state] (e2)[above left of =s1]{};
   
   \node [state] (e3) [right of =s2]{};
   \node [state] (e4) [right of =s5]{};
   \node [state] (e5) [below right  of =s5] {};
   \node [state] (e6) [ below  of =s5]{};
   
\node [state] (e7) [below of =s6]{};
\node [state] (e8)[below of =s7]{};

\node[state, draw=none] (n1)[above of =e1]{};
\node[state, draw=none](n2)[above right of =e1]{};
\node[state, draw=none](n3)[above of =e2]{};

\node[state, draw=none](n4)[below of =e8]{};
\node[state, draw=none](n5)[below right of=e3]{};
\node[state, draw=none](n6)[below of =e3]{};

\path (s1) edge[blue] node {b} (s2);
\path (s1) edge[blue] node {c}  (s3);

\path (s2) edge[blue] node {a} (s4);
\path (s3) edge[blue] node {b} (s6);
\path (s3) edge[blue] node {c} (s7);

\path (s4) edge[blue] node {a} (s5);

\path (e1) edge[blue] node{a} (s1);
\path (e2) edge[blue] node{a} (s1);

\path (s2) edge[blue] node{b} (e3);
\path (s5) edge[blue] node{c} (e4);
\path (s5) edge[blue] node{c} (e5);
\path (s5) edge [blue] node{c} (e6);

\path (s6) edge[blue] node{b} (e7);
\path (s7) edge[blue] node{b} (e8);

\path(n1) edge[blue] node{b}(e1);
\path (n2) edge[blue]node{b}(e1);
\path (n3) edge[blue] node{b}(e2);

\path(e8) edge[blue] node{a}(n4);
\path(e3)edge[blue] node{a}(n5);
\path(e3) edge[blue] node{b}(n6);

 %%%%%%%%%%%%%%%%%%%%%%%%%%%%%%%%%%%%%%%%%

\node[state, fill=mygray] (a1) [left=2.5cm of s1] {24};
    \node[state, fill=mygray] (a2) [below right of =a1]{31};
   \node[state, fill=mygray ](a3)[below left of =a1]{63};
   
   \node[state, fill=mygray](a4)[below of =a2]{9};
   \node[state, fill=mygray](a6)[ left of =a4]{41};
   \node[state, fill=mygray](a7)[left of=a6]{74};
   
   \node[state, fill=mygray] (a5)[below of =a4]{4};

   \node[state] (t1)[ above of =a1]{};

   \node [state] (t6) [ below  of =a5]{};

\node [state] (t8)[below of =a7]{};
\node[state] (t9)[below left  of =a7 ]{};

\node[state, draw=none](m1)[above of =t1]{};
\node[state, draw=none](m3)[below of=t8]{};
\node[state, draw=none](m2)[below left of =t8]{};

\path (a1) edge[purple] node {b} (a2);
\path (a1) edge[purple] node {c}  (a3);

\path (a2) edge[purple] node {a} (a4);
\path (a3) edge[purple] node {b} (a6);
\path (a3) edge[purple] node {c} (a7);

\path (a4) edge[purple] node {a} (a5);

\path (t1) edge[purple] node{a} (a1);

\path (a5) edge[purple] node{c} (t6);

\path (a7) edge[purple] node{a} (t8);
\path (a7) edge[purple] node{a} (t9);
\path(m1) edge[purple]node{b}(t1);

\path(t8) edge[purple] node{c}(m2);
\path(t8)edge[purple] node{a}(m3);

\end{tikzpicture}}
\begin{tikzpicture}[->,>=stealth',shorten >=1pt,auto,node distance=4.2em,semithick]
  \tikzstyle{every state}=[]
  
  \node[state, fill=mygray] (s1) [] {};
    \node[state, fill=mygray] (s2) [below right of =s1]{};
   \node[state, fill=mygray ](s3)[below left of =s1]{};
   
   \node[state, fill=mygray](s4)[below of =s2]{};
   \node[state, fill=mygray](s6)[ left of =s4]{};
   \node[state, fill=mygray](s7)[left of=s6]{};
   
   \node[state, fill=mygray] (s5)[below of =s4]{};

   \node[state] (e1)[above right of =s1]{};
   \node[state] (e2)[ above of =s1]{};
   
   \node [state] (e3) [right of =s2]{};
   \node [state] (e4) [right of =s5]{};
   \node [state] (e5) [below right  of =s5] {};
   \node [state] (e6) [ below  of =s5]{};
   
   \node [state] (t2) [below left of =s5]{};
   
\node [state] (e7) [below of =s6]{};
\node [state] (e8)[below of =s7]{};
\node [state] (t1)[above left of =s1]{};

\node [state] (t3)[below left of =s7]{};
\node [state ](t4) [left of=s7]{};

 \node[state, draw=none](n1)[above of =t1]{};
 \node[state, draw=none](n2)[above of=e2]{};
 \node[state, draw=none](n3)[above of=e1]{};
 \node[state, draw=none](n4)[above right of =e1]{};
 
 \node[state, draw=none](n5)[below of =e3]{};
 \node[state, draw=none](n6)[below right of =e3]{};
 
 \node[state, draw=none](n7)[below of =e8]{};
 \node[state, draw=none](n8)[below left of =t3]{};
 \node[state, draw=none](n9)[below of =t3]{};
 
\path (s1) edge node {b} (s2);
\path (s1) edge node {c}  (s3);

\path (s2) edge node {a} (s4);
\path (s3) edge node {b} (s6);
\path (s3) edge node {c} (s7);

\path (s4) edge node {a} (s5);

\path (e1) edge[blue] node{a} (s1);
\path (e2) edge[blue] node{a} (s1);

\path (s2) edge[blue] node{b} (e3);
\path (s5) edge[blue] node{c} (e4);
\path (s5) edge[blue] node{c} (e5);
\path (s5) edge [blue] node{c} (e6);
\path (s5) edge[purple] node{c} (t2);

\path (s6) edge[blue] node{b} (e7);
\path (s7) edge[blue] node{b} (e8);

\path (t1) edge[purple] node{a} (s1);
\path (s7) edge[purple] node{a} (t3);
\path (s7) edge[purple] node{a} (t4);

\path(n1)edge[purple] node{b} (t1);
\path(n2) edge[blue] node{b} (e2);
\path (n3) edge[blue] node{b} (e1);
\path(n4) edge[blue] node{b} (e1);

\path(e3)edge[blue] node{b}(n5);
\path(e3)edge[blue] node{a}(n6);

\path(e8) edge[blue] node{a}(n7);
\path(t3) edge[purple] node{c}(n8);
\path(t3) edge[purple] node{a}(n9);

\end{tikzpicture}
\caption{Tunneling a block of size $7$ and width $2$ in a Wheeler graph.}
\label{fig:tunnel}
\end{figure}
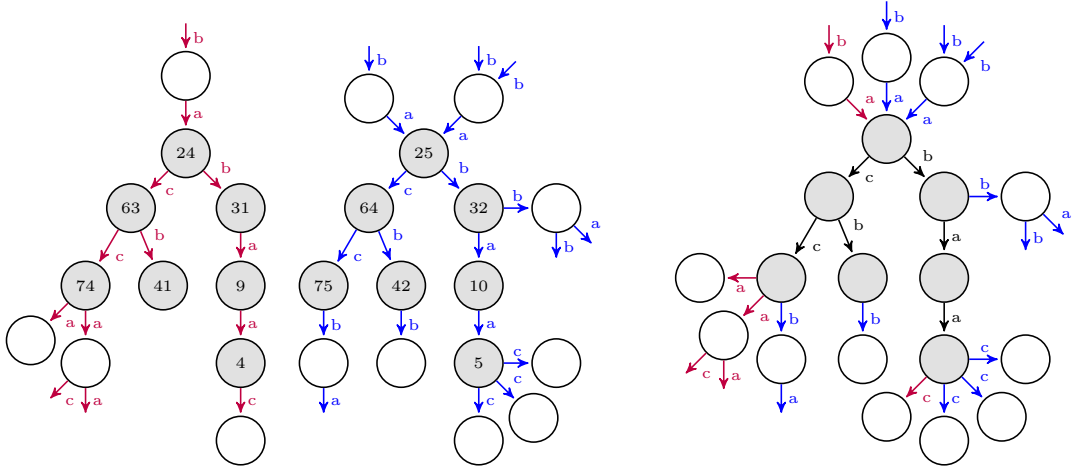

A block $\mathcal{B}=((v_{1,1}, \dots, v_{w,1}), \dots, $ $(v_{1,s}, \dots,v_{w,s}))$, abbreviated $\mathcal{B}=(v_{i,j})_{1 \leq i \leq w, 1 \leq j \leq s}$, is called \emph{maximal in size} if, for any choice of nodes $v_1, \dots, v_{w} \in V$, the sequences of size $s+1$ of $w$-tuples $((v_i)_{1 \leq i \leq w}, (v_{i,1})_{1 \leq i \leq w}, \dots,$ $(v_{i, s})_{1 \leq i \leq w})$ and $((v_{i,1})_{1 \leq i \leq w}, \dots,$ $(v_{i, s})_{1 \leq i \leq w}, (v_i)_{1 \leq i \leq w})$ do not form a block.
The block is called \emph{maximal in width} if, for any choice of nodes $v_1, \dots, v_{s} \in V$, the sequences of $(w+1)$-tuples $((v_j,v_{1,j}, \dots, v_{w,j}))_{1 \leq j \leq s}$ and $((v_{1,j}, \dots, v_{w,j}, v_j))_{1 \leq j \leq s}$ do not form a block.  A block is called \emph{maximal} if it is maximal in both width and size. \\

Let $\mathcal{G}=(V, E, \lambda)$ denote a Wheeler graph containing a maximal block. We then obtain a directed edge-labeled (multi-)graph $\mathcal{G}_t$ from $\mathcal{G}$ as follows:
\begin{itemize}
\item[(i)] We merge the corresponding nodes and edges of the isomorphic subtrees $t_i$, with $1 \leq i \leq w$, in order to obtain a subtree $t$ of $\mathcal{G}$. In particular, for $1 \leq j \leq s$, we collapse the nodes of the $w$-tuple $(v_{1,j}, \dots, v_{w,j})$ to obtain a node $x_j$ of the graph $t$.
The labels of the merged edges coincide and stay the same. 
\item[(ii)] All edges leading from a non-block node $u$ to the root node $v_{i,1}$ of the subtree $t_i$ are redirected to lead to the node $x_1$ of $t$, preserving their labels. 
\item[(iii)] All edges leading from a node $v_{i,j}$ of subgraph $t_i$  to a non-block node $u$ are redirected to leave from the node $x_j$ of $t$, preserving their labels. 
\end{itemize}

Formally, the graph $\mathcal{G}_t=(V_t, E_t, \lambda_t)$ is defined as follows:
The set of nodes of $\mathcal{G}_t$ is defined as $V_t = (V \setminus \{v_{i,j} \mid 1 \leq i \leq w, 1 \leq j \leq s\}) \cup \{x_j \mid 1 \leq j \leq s\}$. We define a function $\varphi: V \rightarrow V_{t}$ mapping a node in $\mathcal{G}$ to its corresponding node in $\mathcal{G}_t$ by
\begin{align*}
\varphi(v) = \begin{cases} v \quad &\text{ if } v \notin \{v_{i,j} \mid 1 \leq i \leq w, 1 \leq j \leq s\} \\
x_j  \quad &\text{ if } v=v_{i,j} \text{ for some integers } 1 \leq i \leq w, 1 \leq j \leq s.
\end{cases}
\end{align*}
The multiset of edges $E_t$ is defined as the difference of multisets $\{(\varphi(u), \varphi(v)) \mid (u,v) \in E\}\setminus \{(\varphi(u), \varphi(v)) \mid (u,v) \in E_i, 2 \leq i \leq w\}$. For every edge $(x,y) \in E_t$, there is a corresponding edge $(u,v) \in E$ such that $(x,y)=(\varphi(u), \varphi(v))$. Thus, we define $\lambda_t((x,y)) = \lambda((u,v))$. This is well-defined because only edges with the same label are merged. We call {\em tunneling} the process of obtaining the graph $\mathcal{G}_t$ from $\mathcal{G}$. See Fig.~\ref{fig:tunnel}.

\begin{lemma}\label{wheelerlemma}
Let $\mathcal{G}=(V, E, \lambda)$ denote a Wheeler graph containing a maximal block $ \mathcal{B}=(v_{i,j})_{1 \leq i \leq w, 1 \leq j \leq s}$ of width $w$ and size $s$ and let $\mathcal{G}_t=(V_t, E_t, \lambda_t)$ denote the graph obtained from $\mathcal{G}$ by tunneling. Then $\mathcal{G}_t$ is a Wheeler graph. 
\end{lemma}

\begin{proof}
%\textcolor{red}{See the longer version of this article [cite].}
To show that $\mathcal{G}_t$ is a Wheeler graph, we must define a Wheeler ordering on $V_t$. As only consecutive nodes of $V$ are merged in order to obtain $\mathcal{G}_t$, this induces a canonical ordering on the nodes of $V_t$: Pick two nodes $x\neq y$ of $V_t$. Let $\varphi^{-1}(x), \varphi^{-1}(y)\subseteq V$ denote the corresponding preimages under $\varphi$. By property (i) of Def.~\ref{def:block}, the nodes of $\varphi^{-1}(x)$ (respectively, $\varphi^{-1}(y)$) are consecutive in Wheeler order. Thus, we either have $u<v$ for every $u \in \varphi^{-1}(x), v \in \varphi^{-1}(y)$, or vice versa. We set $x<y$ in the first case and $x>y$ in the second. This yields an ordering on the nodes of $V_t$ such that for any two nodes $u \neq v$ of $V$, we have $\varphi(u) < \varphi(v) \Rightarrow u < v $ and $u < v \Rightarrow \varphi(u) \leq \varphi(v)$. 

First, take two nodes $x \neq y$ of $\mathcal{G}_t$, such that $x$ has in-degree $0$ and $y$ is of positive in-degree. As in the process of tunneling the in-degree of a node is not decreased, every node $u \in \varphi^{-1}(x)$ is of in-degree $0$ as well. Moreover, as $y$ is of positive in-degree, there is a node $v \in \varphi^{-1}(y)$, such that $v$ is of positive in-degree. As $\mathcal{G}$ is a Wheeler graph, we have $u < v$ and thus $x=\varphi(u) \leq \varphi(v)=y$. As by assumption, $x \neq y$, this yields $x<y$. Therefore, nodes with in-degree $0$ precede those with positive in-degree. 

Now, take two edges $(x,y)$ and $(x', y')$ of $\mathcal{G}$ labeled with $a$ and $a'$, respectively. Without loss of generality, assume $a \preceq a'$. Choose $u \in \varphi^{-1}(x)$, $v \in \varphi^{-1}(y)$, $u' \in \varphi^{-1}(x')$ and $v' \in \varphi^{-1}(y')$, such that $(u,v)$ and $(u', v')$ are edges of $\mathcal{G}$. By definition of $\mathcal{G}_t$, the label on the edge $(u,v)$ (resp., $(u', v')$) is $a$ (resp., $a'$). We then have $(x,y) = (\varphi(u), \varphi(v))$ and $(x', y') = (\varphi(u'), \varphi(v'))$.
Consider the two cases of Def.~\ref{def:wg}: \\ 
(i) Let $a \prec a'$. As $\mathcal{G}$ is a Wheeler graph, we have $v < v'$ and thus $\varphi(v) \leq \varphi(v')$. If $\varphi(v) < \varphi(v')$ we are done, so assume $\varphi(v) = \varphi(v')$. We then have $v=v_{i,j}$ and $v'=v_{k,j}$ for some nodes $v_{i,j} \neq v_{k,j}$ of the block. By properties (ii) - (iv) of Def.~\ref{def:block}, the labels of the incoming edges of $v_{i,j}$ and $v_{k,j}$ are the same, contradicting $a \prec a'$.\\
(ii) Let $a = a'$ and without loss of generality assume $u < u'$. This yields $\varphi(u) \leq \varphi(u')$. As $\mathcal{G}$ is a Wheeler graph, we obtain $v \leq v'$, and thus $\varphi(v) \leq \varphi(v')$. 
\end{proof}

Two blocks $\mathcal{B}_1=(v_{i,j})_{1 \leq i \leq w_1, 1 \leq j \leq s_1}$ and $\mathcal{B}_2=(u_{i,j})_{1\leq i \leq w_2, 1 \leq j \leq s_2}$ of a Wheeler graph $\mathcal{G}=(V, E, \lambda)$ are called \emph{disjoint} if their corresponding node sets $\{v_{i,j}\mid 1 \leq i \leq w_1, 1 \leq j \leq s_1\} \subset V$ and $\{u_{i,j}\mid 1 \leq i \leq w_2, 1 \leq j \leq s_2\} \subset V $ are disjoint. Since, by Lemma \ref{wheelerlemma}, a graph obtained from a Wheeler graph by tunneling is still a Wheeler graph, we can tunnel iteratively with disjoint blocks.

\begin{defi}[Tunneled Graph] \label{def:tunneled}
Let $\mathcal{G}$ be a Wheeler graph containing $k$ pairwise disjoint maximal blocks $\mathcal{B}_1, \dots, \mathcal{B}_k$. The \emph{tunneled graph} $\mathcal{G}_t$ of $\mathcal{G}$ corresponding to those blocks is defined as the Wheeler graph obtained from $\mathcal{G}$ by iteratively tunneling all the blocks $\mathcal{B}_i$, for $1 \leq i \leq k$. Each maximal block $\mathcal{B}_i$ is also called a {\em tunnel}. \qed
\end{defi}

Note that tunnels more general than the tree form given in Def.~\ref{def:block} would break the Wheeler graph rules of Def.~\ref{def:wg} before or after tunneling, except that the nodes of $t_1$ and $t_w$ could be connected with outside nodes, all of Wheeler ranks smaller and larger, respectively, than their corresponding tunnel edges. We could also handle forests, but those can be seen as a set of disjoint tunnels.

\Section{Path searching on a tunneled Wheeler graph}

%With these data structures, we can simulate a search in the original graph $\mathcal G$. We can support the following operations: given the rank $r$ of a node $v$, we can find the number of edges with label $c$ going out from $v$ by computing $$rank_c(L,select_1(O, r+1)-1) - rank_c(L,select_1(O, r)-1) \; ,$$ and we can find the rank of the destination node of the edge with rank $k$ by computing $C[c] + rank_1(I, k)$, where $c$ is the label of the edge.
%($rank_1$ on $I$ and $C$-array).

%We can search for substrings by following the smallest and the largest edges with label $c$ from the current nodes. The ranks of these edges can be found with rank on $L$. In other words:

%Suppose we have an edge $e = (u,v)$ with rank $k$ and we want to take the smallest edge with label $c$ from $v$. First we find the rank of node $v$, which is $r = rank_1(I,k)$. The edges from $v$ are in $L[select_1(L,r)..select_1(L,r+1)]$. By using rank we can find the rank of the first out-edge from $v$ with $c$. 

%Both (a) and (b) are described just before Lemma 4.

Wheeler graphs can be searched for the existence of paths whose concatenated labels yield a given string $P$ \cite{gagie2017wheeler}, generalizing the classical backward search on strings \cite{FM05}. We now show that those searches can also be performed on tunneled Wheeler graphs.

Given $\mathcal G_t$, we can simulate the traversal of $\mathcal G$ as follows. A node $v \in V$ is represented by a pair $\langle \varphi(v), \texttt{off}(v)\rangle $, where $ \varphi(v)$ is the corresponding node in $\mathcal G_t$ as defined in the previous section and $\texttt{off}(v)$ is the \emph{tunnel offset} of $v$. If the node $\varphi(v)$ does not belong to a tunnel, then it must be that $\texttt{off}(v)=1$ and the pair represents just the node $v$. Otherwise, $\varphi(v)$ corresponds to multiple nodes $(v_{i,j})_{1\le i\le w}$ of some tunnel $\mathcal{B}=(v_{i,j})_{1\le i\le w,1\le j\le s}$ in $\mathcal G$ and the tunnel offset represents which of the original nodes we are currently at, in Wheeler rank order. That is, if $v=v_{i,j}$, then $\texttt{off}(v)=i$.

%Similarly, each edge between two nodes in a tunnel corresponds to multiple edges in $\mathcal{G}$. Let $\rho : E \rightarrow E_t$ be a mapping such that $\rho((u,v)) = (\varphi(u), \varphi(v))$. Each edge $e \in E$ is represented by a pair $\langle \rho(e), \texttt{off}(e) \rangle$, where $\texttt{off}(e) = |\{ e' \in E \; | \; \rho(e') = \rho(e) \wedge e' \leq e \} | $.

The idea is that, when our traversal enters a tunnel $\mathcal B$, we remember which original subgraph $t_i$ we actually entered, and use that information to exit the tunnel accordingly. %This will be feasible due to condition (iv) in Def.~\ref{def:block}.
We mark the nodes of $\mathcal G_t$ that are tunnel entrances in a bitvector, and the other tunnel nodes in another bitvector. We then distinguish three cases.

\paragraph{Keeping out of tunnels.} If we are at a pair $\langle i,1\rangle$ not in a tunnel, compute $j$ and $r$ with Eqs.~(\ref{eq:node_to_edge}) and (\ref{eq:edge_to_node}), respectively, and it turns out that $r$ is not marked as a tunnel entrance, then we stay out of any tunnel and our new pair is $\langle r,1\rangle$.

\paragraph{Entering a tunnel.} Assume we are at a pair $\langle i,1\rangle$, where $i=\varphi(v)$ for some non-tunnel node $v$, compute $j$ and $r$ with Eqs.~(\ref{eq:node_to_edge}) and (\ref{eq:edge_to_node}), respectively, and then it turns out that $r = \varphi(u)$ is marked as a tunnel entrance. Then we have entered a tunnel and the new pair must be $\langle r,o\rangle$, for some offset $o$ we have to find out.

The $w$ nodes $(u_p)_{1 \le p \le w} \in V$ that were collapsed to form $\varphi(u)$ are of indegree $0$ within the subgraphs $t_p$, but may receive a number of edges from non-tunnel nodes (indeed, we are traversing one). Since all those edges are labeled by the same symbol $c$, the Wheeler rank of all the sources of edges that lead to $u_p$ must precede the Wheeler ranks of all the sources of edges that lead to $u_{p'}$ for any $1 \le p < p' \le w$.

The problem is, knowing that we are entering by the edge with Wheeler rank $j$, and that the edges that enter into $r$ start at Wheeler rank $select_1(I,r)-r+1$, how to determine the index $o$ of the subgraph $t_o$ we have entered. For this purpose,
we store a bitvector $I'[1..m]$, where $m=|E_t|$, so that $I'[j]=1$ iff the $j$th edge of $E_t$, in Wheeler order, corresponds to the first edge leading to its target in $\mathcal G$. Said another way, $I'$ marks, in the area of $I$ corresponding to the edges that reach $\varphi(u)$, which were the first edges arriving at each copy $u_p$ that was collapsed to form $\varphi(u)$. We can then compute $o = rank_1(I',j)-rank_1(I',select_1(I,r)-r)$.

\paragraph{Moving in a tunnel.} Assume we are at a pair $\langle i,o\rangle$ inside a tunnel, for $i=\varphi(u)$, and want to traverse the $k$th edge labeled $c$ leaving the pair. We then use the formulas given after Eqs.~(\ref{eq:node_to_edge}) and (\ref{eq:edge_to_node}) to compute the first and last edge labeled $c$ leaving node $i$. These form a Wheeler range $[j_1,j_2]$. We then apply Eq.~(\ref{eq:edge_to_node}) from $j=j_1$ to find the first target node $r$. If $r$ is marked as an in-tunnel node, then we know that letter $c$ keeps us inside the tunnel, $j_1=j_2$ (that is, there is exactly one out-edge by letter $c$ from each of the nodes $(u_p)_{1\le p\le w}$ that were collapsed to form $\varphi(u)$, by part (a) of item (v) of Definition \ref{def:block}), and thus our new node is simply $\langle r,o\rangle$.

If, instead, $r$ is not an in-tunnel node, then letter $c$ takes us out of the tunnel, and we must compute the appropriate target node. Each of the nodes $u_s$ may have zero or more outgoing edges labeled $c$. The Wheeler ranks of the edges leaving $u_p$ precede those of the edges leaving $u_{p'}$, for any $1 \le p < p' \le w$. Thus, we use a bitvector $O'[1..m]$ analogous to $I'$, where $O'[j]=1$ iff the $j$th edge of $E_t$ in Wheeler order corresponds to the first edge leaving a node in $\mathcal G$ by some letter $c$. In the range $O'[j_1,j_2]$, then, each $1$ marks the first edge leaving from each $u_p$. The $k$th edge labeled $c$ leaving from $u_o$ is thus $j = select_1(O',rank_1(O',j_1)+o-1)+k-1$. If, for pattern searching, we want the last edge labeled $c$ leaving from $u_o$, this is $j = select_1(O',rank_1(O',j_1)+o)-1$. We then compute the correct target node rank $r$ using Eq.~(\ref{eq:edge_to_node}).

Finally, there are two possibilities. If the new node $r$ is marked as a tunnel entrance, then we have left our original tunnel to enter a new one. We then apply the method described to enter a tunnel from the values $j$ and $r$ we have just computed. Otherwise, $r$ is not a tunnel node and we just return the pair $\langle r',1\rangle$.

\bigskip

Therefore, we can simulate path searching in $\mathcal G$ by using just our representation of $\mathcal G_t$. Given a character $c$ and the Wheeler range $[i,i']$ of a string $S$, we can find the Wheeler range of the string $Sc$ by following the first and last edge labeled with $c$ leaving from $[i,i']$, as described after Eqs.~(\ref{eq:node_to_edge}) and (\ref{eq:edge_to_node}), and operate as described from the corresponding ranks $j$ and $r$. Thus, after $|P|$ steps, we have the Wheeler range of the nodes that can be reached by following the characters in $P$.

\begin{theorem}
We can represent a Wheeler graph $\mathcal G$ with labeled edges from the alphabet $[1..\sigma]$ in $n_t \log \sigma + o(n_t\log\sigma)+O(n_t)$ bits of space, where $n_t$ is the number of edges in a tunneled version of $\mathcal G$, such that we can decide if there exists a path labeled with $P$ in time $O(|P| \log \log_w \sigma)$ in a $w$-bit RAM machine.
\end{theorem}

%We can find the Wheeler ranks of the first and the last edges labelled with $c$ in the range by using rank queries on $L$.

%%%% SEARCHING IN A TUNNELED WHEELER GRAPH %%%%%%%%%%%%%%%%%%%%%%%%%%%%%%%%%%%%

\Section{Wheeler graphs of strings}

We now focus on a particular type of Wheeler graphs, which corresponds to the traditional notion of Burrows-Wheeler Transform (BWT) \cite{BW94} and FM-index \cite{FM05} (only that our arrows go forward in the text, not backwards), and show that the full self-index functionality on strings can still be supported after tunneling. This is close to the original tunneling concept developed by Baier \cite{baier2018undetected}, to which we now add search and traversal capabilities.

\begin{defi}[Wheeler graph of a string]
Let $T$ be a string over an alphabet $A$. The \emph{Wheeler graph of the string} $T$ is defined as $\mathcal{G} = (V, E, \lambda)$ with $V:=\{v_1, \dots, v_{|T|+1}\}$, $E = \{(v_{i}, v_{i+1}) \mid 1 \leq i \leq |T|+1\}$ and $\lambda: E \rightarrow A$ with $\lambda((v_i, v_{i+1})) = T[i]$, where $T[i]$ denotes the $i$th character in the string $T$. \qed
\end{defi}
In other words, the Wheeler graph of a string $T$ is a path of length $|T|+1$, where the $i$th edge is labeled with the $i$th character of $T$. There is exactly one valid Wheeler-ordering, which is given by the colexicographic order of prefixes of $T$, i.e., node $v_i$ comes before node $v_j$ iff the reverse of $T[1..i-1]$ is lexicographically smaller than the reverse of $T[1..j-1]$. There is a close connection to the BWT: the Wheeler order is given by the suffix array of the reverse of $T$ and therefore the $L$-array corresponds to the BWT of the reverse of $T$.

For this special case of Wheeler graphs, Def.~\ref{def:block} simplifies as follows:
A block $\mathcal{B}$ in $\mathcal{G}$ of width $w$ and length $s$ is a sequence of length $s+1$ of $w$-tuples $(v_{1,1}, \dots, v_{w,1}), \dots, $ $(v_{1,s+1}, \dots,v_{w,s+1})$ of pairwise distinct nodes of $\mathcal{G}$ satisfying 
\begin{itemize}
\item[(i)] For  $1 \leq i \leq w-1$ and $1 \leq j \leq s+1$, the immediate successor of the node $v_{i,j}$ with respect to the Wheeler ordering on $V$ is  $v_{i+1,j}$. 
\item[(ii)] For $1 \leq i \leq w$ and $1 \leq j \leq s$, $(v_{i,j}, v_{i,j+1})$ is an edge of $E$.
\item[(iii)] For $1 \leq j  \leq s$, all the edges leading to the nodes in $\{v_{i,j} \mid 1 \leq i \leq w\}$ have the same label.
\end{itemize}

The process of tunneling in a Wheeler graph $\mathcal{G}$ of a string then consists of
collapsing the nodes of each $w$-tuple $(v_{1,j}, \dots, v_{w,j})$ into a single node $x_j$ and collapsing the edges in $\{(v_{i,j}, v_{i,j+1}) \mid 1 \leq i \leq w\}$ into a single edge $(x_j, x_{j+1})$. 
Furthermore, all edges leading to a node $v_{i,1}$ for some integer $1 \leq i \leq w$ are redirected to lead to the node $x_1$ and all edges leaving from a node $v_{i,s}$ for some integer $1 \leq i \leq w$ are redirected to leave from the node $x_s$. The labels of the edges stay the same. Note that we ensure that every path of the tunnel is followed by a non-tunnel node.

Suppose we have the Wheeler graph $\mathcal G = (V,E,\lambda)$ of a string $T$. Denote $|T| = n$. Let $\mathcal G_t = (V_t, E_t, \lambda_t)$ be a tunneled version of $\mathcal G$ with $|V_t| = n_t$. We represent $\mathcal G_t$ with the data structures $L$, $C$, $I$ and $O$ described in the preliminaries. We can do path searches without the bitvectors $I'$ and $O'$ used in the previous section, as in these particular graphs they are all 1s.
We now describe how to implement the operations count, locate and extract, analogous to the operations in a regular FM-index, by using sampling schemes that extend those of the standard FM-index solution.

First, for each tunnel of length at least $\log n_t$ in $\mathcal G_t$, we store a pointer and the distance to the end of the tunnel for every $(\log n_t)$th consecutive node in the tunnel. This information takes $O(n_t)$ bits of space and lets us skip to the end of the tunnel in $O(\log n_t)$ steps by walking forward until the end of the tunnel is found, or until we hit a node with a pointer to the end. The stored distance value tells us how many nodes we have skipped over.

\paragraph{Locating.}
We define a graph $\mathcal G_c$, called the \emph{contracted graph}, that is identical to $\mathcal G_t$ except that tunnels have been contracted into single nodes. Let $\psi : \mathcal G_t \rightarrow \mathcal G_c$ be the mapping such that nodes in a tunnel in $\mathcal G_t$ map to the corresponding contracted node in $\mathcal G_c$. Take the path $Q = (q_1, \ldots, q_n)$ of all nodes in $\mathcal G $ in the order of the path from the source to the sink. Let $Q_c = ((\psi \circ \varphi)(q_1),\ldots,(\psi \circ \varphi)(q_n))$ be the corresponding path in $\mathcal G_c$. Let $Q_c'$ be the same sequence as $Q_c$ except that every run of the same node is contracted to length 1. Note that this path traverses all edges of $\mathcal G_c$ exactly once, i.e., it is an Eulerian path. We store a sample for every $(\log n)$th node in the path on $Q_c'$, except that if a node represents a contracted tunnel, we sample the next node (our definition of blocks guarantees that the next node is not in a tunnel). The value associated with the sample is the text position corresponding to the node. This takes space $O(n_c)$, where $n_c$ is the number of nodes in $\mathcal G_c$.

We can then locate the text position of a node by walking to the next sample in text order, using at most $\log n$ graph traversal operations in $\mathcal G_t$. 
 In this walk we may have to skip tunnels, which is done in $O(\log n_t)$ time using their stored pointers when necessary. In the end, we subtract the travelled distance from the text position of the sampled node to get the text position of the original node. We can view such a search as a walk in $\mathcal G_c$, where traversing a contracted-tunnel node takes $O(\log n_t)$ graph traversal steps and traversing a non-tunneled node takes just one. 
The worst case time, dominated by the time to traverse tunnels, is $O(\log n \log n_t)$ steps.

%\begin{lemma}
%Given the above data structures, the Wheeler-rank and the tunnel offset of a node in $\mathcal G_t$, it is possible to compute the text-position of the node in time $O(\log n_c \log n_t \log \sigma)$. 
%\end{lemma}

%\begin{proof}
% We walk forward in the graph until we hit a sampled node such that if we are in a tunnel, we skip to the end of the tunnel in $O(\log n_t)$ time. This can be viewed as a walk in $\mathcal G_c$ and by construction we need to take at most $2 \log n_c$ steps in $\mathcal G_c$ until we hit a sampled node. The worst-case time is dominated by traversing the tunnels. The path will contain at most $\log n_c$ tunnels, each of which can be traversed through in $\log n_t$ steps, so the total time taken is $O(\log n_c \log n_t)$.
%\end{proof}

\paragraph{Counting.}
Efficiently counting the number of occurrences of a pattern given its Wheeler-range requires a sampling structure different from that used for locating. The Wheeler-range could span many tunnels, whose widths are not immediately available. Let us define $w(v)$ as the width of the tunnel $v$ belongs to, or $1$ if $v$ does not belong to a tunnel. We can then afford to sample the cumulative sum of the values $w(v)$ for all the nodes $v$ up Wheeler-rank $k$ for every $k$ multiple of $\log n_t$ nodes, using $O(n_t)$ bits of space. Within this space we can also mark which nodes belong to tunnels.

This allows us to compute the sum of values $w(v)$ for any Wheeler range with endpoints that are multiples of $\log n_t$, which leaves us to compute the width of only $O(\log n_t)$ nodes at the ends of the range. For these nodes, we add $1$ if they are not in a tunnel; otherwise we go to the end of the tunnel using the stored pointers and compute the width of the tunnel by looking at the out-degree of the exit of the tunnel. The total counting time is then $O(\log^2 n_t)$ graph traversal steps.

%\begin{lemma}
%Given the above data structures and representations of the first and the last node in the range of a pattern, it is possible to count the number of occurrences of the pattern in time $O(\log^2 n_t \log \sigma)$. 
%\end{lemma}
%\begin{proof}

%We can query the sum of tunnel widths for any Wheeler range with endpoints that are multiples of $\log n_t$, which leaves us to compute the width of only $O(\log n_t)$ nodes at the ends of the range. For the rest of the nodes we just go to the end of the tunnel using the previously defined data structure and compute the width of the tunnel by looking at the out-degree of the exit of the tunnel. The total time is then $O(\log^2 n_t)$.
%\end{proof}

\paragraph{Extracting.}
To extract characters from $T$, we use a copy of the samples $\langle$Wheeler rank of graph node $v$, text position of node $v\rangle$ we store for locating, but sorted by text position. 
% another different sampling structure. We store a pair with the text position and the Wheeler rank of every $(\log n)$th node in the Eulerian path $Q_c'$, except that if the sampled node would be a contracted tunnel, we sample the next node on the path. 
Also, at the end of every tunnel of length at least $\log n_t$, we store backpointers to the nodes storing pointers to the end of the tunnel.

%Also at every $\log n_t$ position in every tunnel in $\mathcal G_t$ of length at least $\log n_t$, we store a pointer to the node $\log n_t$ nodes ahead of that node in the tunnel.

Suppose we want to extract $T[i..j]$. If we know the node $u$ of $\mathcal G_t$ representing position $i$, we can simply walk forward from that node to find the $j-i+1$ desired characters by accessing $L$ at each position. Therefore it is enough to show how to find the node $u$. We binary search our sample pairs to find the Wheeler rank of the closest sample before text position $i$. This sample is at most $\log n$ nodes away (in $\mathcal G_c$) from $u$, so we can reach $u$ in $O(\log n)$ steps in $\mathcal G_c$, or equivalently, $O(\log n\log n_t)$ steps in $\mathcal G_t$. Note, however, that our target node $u$ might be in a tunnel. If the tunnel is of length less than $\log n_t$, we walk towards it normally. If $u$ is inside a tunnel of length at least $\log n_t$, instead, we use its pointers to skip to the end of the tunnel, and from there take the backpointer to the nearest position before $u$; we then walk the (at most) $\log n_t$ nodes until reaching $u$.

The time needed to reach $u$ is again dominated by the time to skip over and within tunnels, so the total time complexity is $O(\log n \log n_t)$ graph traversal steps.

\bigskip

We note that, in all cases, our graph traversal steps are of a particular form, because all the edges leaving from the current node are labeled by the same symbol. That is, the $c$ in Eq.~(\ref{eq:node_to_edge}) is always $L[select_1(O,i)]$. This particular form of $rank$ is called {\em partial rank} and it can be implemented in constant time using $o(|L|\log|A|)$ further bits \cite[Lem.~2]{BNtalg13}. The following theorem summarizes the results in this section.

%\begin{lemma}
%Extract in time $O(\log^2 n_t \log \sigma)$.
%\end{lemma}

%\begin{proof}
%Suppose we want to extract characters $T[i], \ldots, T[j]$. If we know the node representing position $i$, we can simply walk forward from that node to find the desired characters. Therefore it is enough to show that we can find position $i$. We use our data structures to jump to the closest sampled node that has text position at most $i$ \footnote{How exactly?}. Then we walk forward, skipping over tunnels in $\log n_t$ time, and using the sampled forward pointers inside tunnels to find the exact node if it is inside a tunnel \footnote{Also probably need to store the length of a tunnel at the start of every tunnel}.
%\end{proof}

\begin{theorem}
We can store a text $T[1..n]$ over alphabet $[1..\sigma]$ in $n_t \log \sigma + o(n_t\log\sigma)+O(n_t)$ bits of space, such that in a $w$-bit RAM machine we can decide the existence of any pattern $P$ in time $O(|P| \log\log_w \sigma)$, and then report the text position of any occurrence in time $O(\log n \log n_t)$ or count the number of occurrences in time $O(\log^2 n_t)$. We can also extract any $k$ consecutive characters of $T$ in time $O(k + \log n \log n_t)$, where $n_t \le n$ is the number of nodes in a tunneled Wheeler graph of $T$.
\end{theorem}

\Section{Future work}

Open problems are: How to find the optimal blocks that minimize space? Can we still support path searching if blocks are overlapping? Can the $O(\log^2 n)$ times of counting, locating, and extracting be reduced to $O(\log n)$, as in the basic sampling scheme on non-tunneled BWTs? How to extend those operations to more complex graphs, like trees? And can we count paths instead of path endpoints? 

\paragraph{Acknowledgements.}

Funded in part by  EU's Horizon 2020 research and innovation
programme under  Marie Sk{\l}odowska-Curie grant agreement No 690941
(project BIRDS). T.G.\ and G.N.\ partially funded with Basal Funds FB0001, Conicyt, Chile. T.G.\ partly funded by Fondecyt grant 1171058.  L.S.B.\  supported by the EU project 731143 - CID and the DFG project LO748/10-1 (QUANT-KOMP).
Operation extract was designed in StringMasters 2018. We thank Uwe Baier for helpful discussions.

\Section{References}
\bibliographystyle{IEEEbib}
\bibliography{refs}

\end{document}